\documentclass[conference]{IEEEtran}

\usepackage{cite}

\usepackage{graphicx,color,epsfig,rotating,subfigure}
\usepackage{amsfonts,amsmath,amssymb}
\usepackage{algorithm}
\usepackage{caption}
\usepackage{subfigure}
\usepackage{algpseudocode}
\usepackage{placeins}
\usepackage{psfrag, graphicx}

\graphicspath{{figures/}}
\usepackage{booktabs}
\usepackage{longtable}
\usepackage{tabu}
\usepackage{multirow}
\usepackage{multicol}
\usepackage{accents}

\newtheorem{definition}{Definition}
\newtheorem{theorem}{Theorem}

\newtheorem{corollary}{Corollary}

\newtheorem{proof}{Proof}

\newtheorem{remark}{Remark}

\newtheorem{example}{Example}

\def\beq{\begin{equation}}
\def\eeq{\end{equation}}
\def\beqa{\begin{eqnarray}}
\def\eeqa{\end{eqnarray}}
\def\beqan{\begin{eqnarray*}}
\def\eeqan{\end{eqnarray*}}

\def\EE{{\mathbb{E}}}
\def\PP{{\mathbb{P}}}

\def\Wsf{ {\sf W}}
\def\Xsf{ {\sf X}}

\def\Usf{ {\sf U}}
\def\Vsf{ {\sf V}}

\def\Xc{\mathcal{X}}
\def\Yc{\mathcal{Y}}
\def\Uc{\mathcal{U}}
\def\Dc{\mathcal{D}}

\def\dbf{\mathbf{d}}

\def\Rstar{R^*(M)}          
\def\RGWstar{R_{GW}^*(M)}   
\def\RGWstarU{R_{GW}^*(M,\Usf)}   
\def\Rach{R_{GW}^{UB}}
\def\Rlb{R^{LB}(M)}         
\def\RGWlb{R_{GW}^{LB}(M)}  
\def\regop{\mathfrak R^{*}}         
\def\regopGW{\mathfrak R^{*}_{GW}}         
\def\GWregion{\mathfrak R_{GW}}   
\def\private{\rho}          
\def\RGWUlb{R_{GW}^{LB}(M,\Usf)}  
\def\GWregion{\mathfrak S_{GW}}   

\IEEEoverridecommandlockouts
\begin{document}

\title{Rate-Memory Trade-off for the Two-User Broadcast Caching Network with Correlated Sources}
\author{Parisa Hassanzadeh, Antonia M. Tulino, Jaime Llorca, Elza Erkip
\thanks{This work has been supported by NSF  grant \#1619129.}
\thanks{P. Hassanzadeh  and  E. Erkip are with the ECE Department of New York University, Brooklyn, NY. Email: \{ph990, elza\}@nyu.edu}
\thanks{J. Llorca  and A. Tulino are with Bell Labs, Nokia, Holmdel, NJ, USA. Email:  \{jaime.llorca, a.tulino\}@nokia-bell-labs.com}
\thanks{A. Tulino is with the DIETI, University of Naples Federico II, Italy. Email:  \{antoniamaria.tulino\}@unina.it}
}

\maketitle

\begin{abstract}
This paper studies the fundamental limits of caching in a network with two receivers and two files generated by a two-component discrete memoryless source with arbitrary joint distribution. 
Each receiver is equipped with a cache of equal capacity, and the requested files are delivered over a shared error-free broadcast link.
First, a lower bound on the optimal peak rate-memory trade-off 
is provided. Then, in order to leverage the correlation among the library files to alleviate the load over the shared link,
a two-step correlation-aware cache-aided coded multicast (CACM) scheme is proposed.
The first step uses Gray-Wyner source coding 
to 
represent the library 
via one common and two private descriptions, 
such that a second correlation-unaware multiple-request CACM step can exploit the additional coded multicast opportunities that arise. 
It is shown that the rate achieved by the proposed two-step scheme matches the lower bound for a significant memory regime  and it is within half of the conditional entropy for all other memory values.
\end{abstract}

\section{Introduction and Setup}~\label{sec:Introduction}

The use of caching at the wireless network edge 
has emerged as a promising approach to efficiently increase the capacity of wireless access networks. 
There have been extensive studies characterizing
the fundamental rate-memory trade-off in a broadcast caching network with a library composed of independent content \cite{maddah14fundamental,ji15order,ji2015caching}.
More recently, \cite{timo2016rate,ISTC2016,ITW2016} have studied the rate-memory trade-off when delivering correlated files.
In \cite{timo2016rate}, the authors consider a single-receiver multiple-file network with lossy reconstructions and characterize the trade-offs between rate, cache capacity, and reconstruction distortions. They extend the analysis to a scenario with two receivers and one cache, in which local caching gains 
can be explored. The works in \cite{ISTC2016} and \cite{ITW2016} consider a setting with an arbitrary number of files and receivers each having a cache. A practical correlation-aware 
scheme is proposed in \cite{ISTC2016}, in which content is cached according to both the popularity of files and their correlation with the rest of the library. Then, the requested content is delivered via a multicast codeword composed of compressed versions of the requested files. 
Alternatively, the work in \cite{ITW2016} 
addresses the dependency among content files by first compressing the correlated library, 
which is then treated as a library of independent files by a
conventional cache-aided coded multicast scheme.

In this paper, by focusing on a two-user two-file setting, we are able to characterize the optimal peak rate-memory trade-off of the broadcast caching network with correlated sources. 
To this end, we first provide a lower bound on the optimal peak rate-memory trade-off, which is derived using a cut-set argument on the corresponding cache-demand augmented graph \cite{llorca2013network}.
The lower bound improves the best known bound for correlated sources given in \cite{lim2016information}, and when particularized to independent sources, matches the corresponding best known bound derived in 
\cite{ghasemi2015improved}.
We then propose a two-step scheme, in which the source files are first encoded based on the Gray-Wyner network \cite{gray1974source}, and in the second step, they are cached and delivered through a multiple-request correlation-unaware cache-aided coded multicast scheme.
In the rest of the paper, we discuss the optimality of the proposed two-step scheme by characterizing a lower bound on the rate-memory trade-off for this class of schemes, and comparing it with the lower bound on the optimal 
trade-off. We identify the set of operating points in the Gray-Wyner region \cite{gray1974source,wyner1975common}, for which a two-step scheme is optimal over a range of cache capacities, and  approximates the optimal rate to within half of the conditional entropy for all cache sizes. 

The paper is organized as follows. Sec. \ref{sec:Problem Formulation} presents the information-theoretic problem formulation. In Sec. \ref{sec: genei scheme},
we introduce a class of 
two-step schemes based on the Gray-Wyner network. The lower bounds 
for the optimal 
and 
two-step schemes 
are provided in Sec. \ref{sec:Lower bound}, and 
later used to establish the optimality of an achievable two-step scheme proposed in Sec.~\ref{sec:achievable scheme}.
After analyzing an illustrative example in Sec.~\ref{sec:Order Optimality}, the paper is concluded in Sec.~\ref{sec:Conclusions}.

\section{Network Model and Problem Formulation}\label{sec:Problem Formulation}
We consider a broadcast caching network composed of one sender 
with access to a library of two files generated by a $2$-component discrete memoryless source (2-DMS). 
The 2-DMS model $(\Xc_1 \times \Xc_2, \, p(x_1, x_2))$ 
consists of two finite alphabets $\Xc_1,\Xc_2$ and a joint pmf  $p(x_1,x_2)$ over $\Xc_1 \times \Xc_2$.
The 2-DMS generates and i.i.d. random process $\{\Xsf_{1i},\Xsf_{2i}\}$ with $(\Xsf_1, \Xsf_2)\sim p(x_1,x_2)$.
For a block length $n$, the two library files are represented by sequences $\Xsf_1^n = (\Xsf_{11},\dots,\Xsf_{1n})$ and $\Xsf_2^n = (\Xsf_{21},\dots,\Xsf_{2n})$, respectively, where $\Xsf_1^n \in \Xc^n_1$ and $\Xsf_2^n \in \Xc^n_2$.
The sender communicates with two receivers 
$r_1$ and $r_2$ over a shared error-free broadcast link. Each receiver is equipped with a cache of size $nM$ bits,
where $M$ denotes the (normalized) cache capacity.

We assume that the system operates in two phases: 
a {\em caching phase} and a {\em delivery phase}. During the {caching phase}, which takes place at off-peak hours when network resources are abundant, the receiver caches are filled with functions of the library files, 
such that during the {delivery phase}, when receiver demands are revealed and resources are limited, the sender broadcasts the shortest possible codeword that allows each receiver to losslessy recover its requested file.
We refer to the overall scheme as a {\em cache-aided coded multicast scheme} (CACM).
Given a realization of the library, $\{X_1^n, X_2^n\}$, a CACM scheme consists of the following components:
\begin{itemize}
\item {\textbf{Cache Encoder:}} During the caching phase, the cache encoder designs the cache content of receiver $r_i$ using a mapping
 $f^{\mathfrak C}_{r_i}:
    \Xc_1^n\times\Xc_2^n \rightarrow [1: 2^{nM})$.
 The cache configuration of  receiver $r_i$  is denoted  by $Z_{r_i} = f^{\mathfrak  C}_{r_i}(X_1^n, X_2^n)$.

\item{\textbf{Multicast Encoder:}} During the delivery phase, each receiver requests a file from the library. 
The demand realization, denoted by $\dbf = (d_{r_1},d_{r_2}) \in \Dc \equiv \{1,2\}^2$, where $d_{r_i} \in\{1,2\}$ denotes the index of the file requested by receiver $r_i$, is revealed to the sender,
which then uses a fixed-to-variable mapping $f^{\mathfrak M}:{\mathcal D}  \times [1: 2^{nM})\times [1: 2^{nM})\times \Xc_1^n\times \Xc_2^n \rightarrow \Yc^\star$ to generate and transmit a multicast codeword $Y_{\dbf} = f^{\mathfrak  M}(\dbf,\{Z_{r_1},Z_{r_2}\}, \{X_1^n,X_2^n\})$ over the shared link.\footnote{We use $\star$ to indicate variable length.}
The codeword $Y_{\dbf}$ is designed for each demand realization according to the cache content, library files, and joint distribution $p(x_1,x_2)$, to enable almost-lossless reconstruction of the requested files.

\item{\textbf{Multicast Decoders:}} Each receiver $r_i$ uses a mapping $g^{\mathfrak  M}_{r_i} : \Dc \times \Yc^\star \times [1: 2^{nM}) \rightarrow \Xc_{d_{r_i}}^n$ to recover its requested file,
    $X_{d_{r_i}}^n$, using the received multicast codeword and its cache content as $\widehat{X}_{d_{r_i}}^n = g^{\mathfrak  M}_{r_i} (\dbf, Y_{\dbf},Z_{d_{r_i}})$.

\end{itemize}

The worst-case probability of error of a CACM scheme is given by 
\begin{align} \label{perr}
& P_e^{(n)} = \max_{\dbf} \;\max_{r_i}\;  \PP\left(\widehat{X}_{{d_{r_i}}}^n  \neq X_{{d_{r_i}}}^n \right).
\end{align}
In this paper, we focus on the peak multicast rate, corresponding to the worst-case demand,
\begin{equation} \label{average-rate}
R^{(n)} =  \max_{\dbf} \; \frac{\EE[L(Y_{\dbf})]}{n},
\end{equation}
where $L(Y)$ denotes the length (in bits) of the multicast codeword $Y$, and the expectation is over the library files. 

\begin{definition} \label{def:achievable-rate}
A peak rate-memory pair $(R,M)$ is {\em achievable} if there exists a sequence of CACM schemes for cache capacity $M$ and
increasing block length $n$, such that $\lim_{n \rightarrow \infty} P_e^{(n)} = 0 \notag$, and $\limsup_{n \rightarrow \infty}
R^{(n)} \leq  R$.
\end{definition}
\begin{definition} \label{def:infimum-rate}
The peak rate-memory region, $\regop$, is the closure of the set of achievable peak rate-memory pairs $(R,M)$, and the optimal peak rate-memory function is 
$$\Rstar= \inf \{R:  (R,M) \in  \regop\}.$$
\end{definition}

\section{Gray-Wyner CACM  Scheme }\label{sec: genei scheme}

In this section, we describe a class of schemes based on a two-step lossless source coding setup, as depicted in Fig.~\ref{fig:caching schemes}. The first step involves compressing the library via Gray-Wyner 
source coding,  
and the second step is a correlation-unaware multiple-request CACM scheme. 
We refer to this scheme as {\em Gray-Wyner Cache-Aided Coded Multicast}  (GW-CACM).

Gray-Wyner  source coding, depicted in Fig.  \ref{fig:caching schemes}, is a distributed lossless source coding setup in which a 2-DMS $(\Xsf_1, \Xsf_2 )$ is represented by three descriptions $ \{\Wsf_0, \Wsf_1, \Wsf_2\}$, where $\Wsf_0 \in [1:2^{nR_0})$, $\Wsf_1 \in [1:2^{nR_1})$, and $\Wsf_2 \in [1:2^{nR_2})$. File ${d_{r_1}}$ can be losslessly recovered from descriptions $(\Wsf_0, \Wsf_{d_{r_1}})$,  and file $d_{r_2}$ can be losslessly recovered from descriptions
 $(\Wsf_0, \Wsf_{d_{r_2}})$, both asymptotically, as $n \rightarrow \infty$.

As shown in \cite{wyner1975common}, the Gray-Wyner rate region is the closure of the union  over $\Usf$ of $\GWregion(\Usf)$, where
$\GWregion(\Usf)$ denotes
the set of rate triplets $(R_0, R_1, R_2)$ such that
\begin{eqnarray}
\label{eq:gwregion}
R_0 & \geq  & I(\Xsf_1, \Xsf_2;\Usf), \label{eq:gwregion1} \\
R_1  & \geq & H(\Xsf_1|\Usf),   \label{eq:gwregion2} \\
R_2 & \geq  & H(\Xsf_2|\Usf),   \label{eq:gwregion3}
\end{eqnarray}
given a conditional {pmf} $p(u|x_1, x_2)$ with $|\Uc| \leq |\Xc_1|.|\Xc_2| + 2$.
 
For a given  $\Usf$ and  a rate triplet $(R_0, R_1, R_2) \in \GWregion(\Usf)$,  a
 GW-CACM scheme consists of:

\begin{figure}[t!] \centering
 \includegraphics[width=1\linewidth]{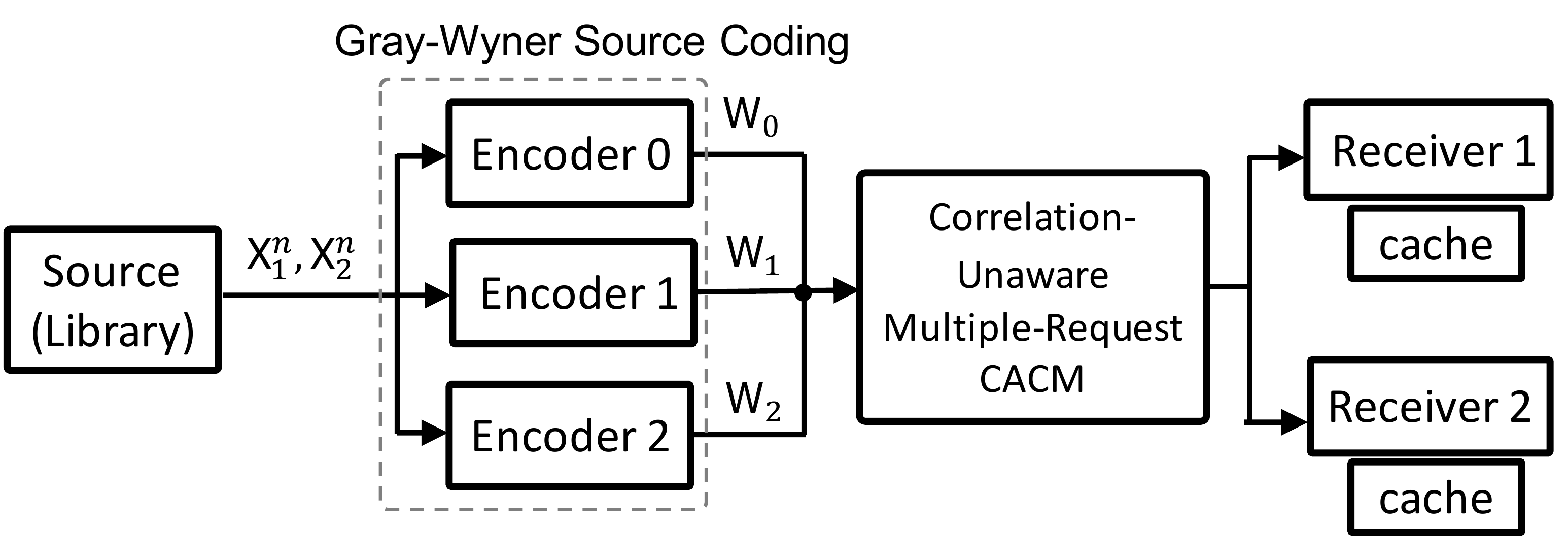}
\caption{Gray-Wyner CACM scheme, composed of a first Gray-Wyner source coding step, and a second correlation-unaware multiple-request CACM step.}
\label{fig:caching schemes}
\end{figure}

 \begin{itemize}
\item {\bf Gray-Wyner Encoder}: Given a library realization $\{X_1^{n}$, $X_2^{n}\}$,
the Gray-Wyner encoder at the sender
    computes three descriptions $ \{\Wsf_0, \Wsf_1, \Wsf_2\}$ using a mapping
    ${f}^{GW}: \Xc_1^{n}\times\Xc_2^{n} \rightarrow [1:2^{nR_0})\times[1:2^{nR_1})\times[1:2^{nR_2})$.

\item {\bf Correlation-Unaware Cache Encoder}: Given the descriptions $\{\Wsf_0, \Wsf_1, \Wsf_2\}$, the cache
    encoder at the sender computes the Gray-Wyner based cache contents $$Z_{r_i} =  f_{r_i}^{{\mathfrak C_{GW}}}
    (\Wsf_0, \Wsf_1 , \Wsf_2), \quad  r_i \in \{1,2\}.$$

\item {\bf Correlation-Unaware Multicast Encoder}: For any demand realization $\dbf$ revealed to the sender, the Gray-Wyner based multicast encoder generates and transmits the multicast codeword
$$Y_\dbf^{{{GW}}} = f^{{\mathfrak M_{GW}}}(\dbf,\{Z_{r_1},Z_{r_2}\}, \{\Wsf_0, \Wsf_1, \Wsf_2\}).$$

\item {\bf Multicast Decoder}: Receiver $r_i$ decodes the descriptions corresponding to its requested file as
$$\{\widehat{\Wsf}_0, \widehat{\Wsf}_{d_{r_i}}\} = g^{{\mathfrak M_{GW}}}_{r_i}(\dbf, Y_\dbf^{{{GW}}}, Z_i).$$

\item {\bf Gray-Wyner Decoder}: Receiver $r_i$ decodes its requested file using the descriptions recovered by the multicast
    decoder, via a mapping $g^{GW}_{r_i}:  [1:2^{nR_0})\times[1:2^{nR_{d_i}}) \rightarrow \Xc_i^n$, as
    $$\widehat{X}_{d_i}^{n} = g^{GW}_{r_i}(\widehat{\Wsf}_0, \widehat{\Wsf}_{d_{r_i}}).$$
\end{itemize}
 
Notice that for the class of GW-CACM schemes, since $(R_0, R_1, R_2) \in \GWregion(\Usf)$ and $(\Wsf_0, \Wsf_{d_{r_i}})$ is a Gray-Wyner description of $X_{{d_{r_i}}}^n$, in order to have
$\lim_{n \rightarrow \infty}  P_e^{(n)}= 0$, with $P_e^{(n)}$ as defined in \eqref{perr},
we only need
\begin{align}
\lim_{n \rightarrow \infty} \max_{\dbf} \;\max_{r_i}\;  \PP\Big( (\widehat{\Wsf}_0, \widehat{\Wsf}_{d_{r_i}}) \neq ({\Wsf}_0, {\Wsf}_{d_{r_i}} )  \Big) = 0. \notag
\end{align}
 
As in \eqref{average-rate},  the peak GW-CACM multicast rate is
\begin{equation} \label{average-rate-GW}
R_{GW}^{(n)}(R_0,R_1,R_2) = \max_{\dbf} \;  \frac{\EE[L(Y^{GW}_{\dbf})]}{n},
\end{equation}
where we explicitly show the dependence on $(R_0,R_1,R_2)$.

In line with Definitions \ref{def:achievable-rate} and \ref{def:infimum-rate}, for a given $\Usf$, the peak $\Usf$-rate-memory region for the class of GW-CACM schemes, $\regopGW(\Usf)$, is defined as the closure of the union of all the achievable pairs $\Big(R_{GW}^{(n)}(R_0,R_1,R_2),M\Big)$ with $(R_0, R_1, R_2)$$ \,\in \,$$\GWregion(\Usf)$. Analogously, the peak $\Usf$-rate-memory function of GW-CACM, $\RGWstarU$, is defined as
$\RGWstarU= \inf \{R:  (R,M) \in  \regopGW(\Usf)\}$.
 
 We remark that $\RGWstarU$ is the rate achieved by a GW-CACM scheme with the Gray-Wyner encoder operating at the boundary of the region
 $\GWregion(\Usf)$. Finally, optimizing over the choice of $\Usf$, we obtain the peak rate-memory region, $\regopGW$,  and  the peak rate-memory function, $\RGWstar$, as
$$\regopGW = cl\Big\{ \bigcup \,\regopGW(\Usf) \Big\}, \;\RGWstar= \inf\RGWstarU,$$
where $cl\{ \mathcal S \}$ denotes  the closure of $ \mathcal S$, and
the union and infimum are over all choices of $\Usf$ with $|\Uc| \leq |\Xc_1|.|\Xc_2| + 2.$

\section{Lower bounds}~\label{sec:Lower bound}
In this section, we provide lower bounds for  $\Rstar$, the optimal peak rate-memory function, and  $\RGWstarU$, the  peak $\Usf$-rate-memory function of GW-CACM for a given $\Usf$.  The latter bound can also be used to obtain a lower bound for $\RGWstar$. We then  investigate conditions on the cache capacity $M$ under which the lower bounds for $\Rstar$ and $\RGWstar$  meet. These conditions 
are then used in Section~\ref{sec:achievable scheme} 
in order to establish the optimality of GW-CACM,  and quantify 
the rate gap from the lower bound 
as a function of the cache capacity $M$.

\subsection{Lower bound on $\Rstar$}
\begin{theorem}\label{thm:LowerBound}
For a broadcast caching network with two receivers, cache capacity $M$, and a library composed of two files with joint distribution $p(x_1,x_2)$,
a lower bound on $\Rstar$, the optimal peak rate-memory function, is given by
\begin{align}
&\Rlb  =  \inf \Big \{ R:   \notag\\
&\qquad R \geq H(\Xsf_1, \Xsf_2)- 2M,\notag\\
&\qquad R \geq \frac{1}{2}\Big(H(\Xsf_1, \Xsf_2)-M\Big), \notag\\
&\qquad R \geq \frac{1}{2}\Big( H(\Xsf_1, \Xsf_2) +   \max\Big\{H(\Xsf_1), H(\Xsf_2) \Big\}\Big) -M \Big \}. \notag
\end{align}
\end{theorem}
\begin{proof}
Theorem \ref{thm:LowerBound} follows from combining cut-set bounds on $i$) 
the cache-demand-augmented graph, 
 and $ii$)  
the time-replication of the cache-demand-augmented graph as described in \cite{llorca2013network,llorca2013minimum}.
\end{proof}

\begin{remark} 
The outer bound in Theorem \ref{thm:LowerBound}  
improves the best known bound for correlated sources given in \cite[Theorem 2]{lim2016information}, and when particularized to independent sources, matches the corresponding best known bound derived in \cite{ghasemi2015improved}. 
\end{remark}

\subsection{Lower bounds on $\RGWstarU$ and $\RGWstar$}
\begin{theorem}\label{thm:LowerBoundGW}
For a given $\Usf$, a lower bound on $\RGWstarU$, the peak $\Usf$-rate-memory function of the GW-CACM scheme, is given by

\vspace{-0.2cm}
\begin{align}
&\RGWUlb = \inf \Big \{ R: \notag\\
& \qquad\quad R \geq I(\Xsf_1, \Xsf_2;\Usf) + H(\Xsf_1|\Usf) + H(\Xsf_2|\Usf) - 2M , \notag\\
& \qquad\quad R \geq \frac{1}{2}\Big (I(\Xsf_1, \Xsf_2;\Usf) + H(\Xsf_1|\Usf) + H(\Xsf_2|\Usf) -M \Big), \notag\\
& \qquad\quad R \geq  I(\Xsf_1, \Xsf_2;\Usf) + H(\Xsf_1|\Usf) + \frac{1}{2} H(\Xsf_2|\Usf) -M,\notag\\
& \qquad\quad R \geq  I(\Xsf_1, \Xsf_2;\Usf) + \frac{1}{2}H(\Xsf_1|\Usf) + H(\Xsf_2|\Usf) -M \notag
\Big \}.  
\end{align}
\end{theorem}
\begin{proof}
The proof is similar to that of Theorem  \ref{thm:LowerBound}; now applied to Gray-Wyner descriptions at rates  $(R_0, R_1, R_2)$$ \,\in \,$$\GWregion(\Usf)$.
\end{proof}
 
\begin{corollary}
A lower bound on $\RGWstar$, the peak rate-memory function of GW-CACM, is given by
 $$\RGWlb = \inf\,\RGWUlb,  
 $$
 \end{corollary}
 where the infimum is over all choices of $\Usf$ with $|\Uc| \leq |\Xc_1|.|\Xc_2| + 2.$

\subsection{Where $\RGWlb$  and $\Rlb$ meet}
By comparing the lower bounds in Theorems~\ref{thm:LowerBound} and \ref{thm:LowerBoundGW}, it is easy to see that  $\Rlb \leq \RGWUlb$, and hence, $\Rlb \leq \RGWlb$. In the following, we derive conditions under which  $\Rlb = \RGWlb$.

\begin{theorem}\label{thm:lower bounds coincide}
Let
\begin{align}
& M_1 \triangleq  \max_{X_1 - U - X_2} \frac{1}{2}  \min \Big\{H(\Xsf_1|\Usf),H(\Xsf_2|\Usf) \Big\}. \notag  
\end{align}
 
 Then, for  $M\in [ 0, \, M_1] \cup [ H(\Xsf_1, \Xsf_2)- 2 M_1 , \,  H(\Xsf_1, \Xsf_2) ] $,  we have $\RGWlb =\Rlb$.

 \end{theorem}
\begin{proof}
Theorem \ref{thm:lower bounds coincide} follows from comparing $\Rlb$ with $\RGWUlb$ for a given $\Usf$, over different regions of memory $M$. It is observed that when 
\begin{align} 
M &\in \Big[ 0, \,\frac{1}{2}  \min \Big\{H(\Xsf_1|\Usf),H(\Xsf_2|\Usf) \Big\}\Big] \bigcup \notag\\
&\Big[H(\Xsf_1, \Xsf_2)- \min \Big\{H(\Xsf_1|\Usf),H(\Xsf_2|\Usf) \Big\} , \,  H(\Xsf_1, \Xsf_2) \Big], \notag
\end{align}
$\RGWUlb-\Rlb$ is independent from the cache capacity $M$, and becomes zero when $I(\Xsf_1, \Xsf_2; \Usf)+ H(\Xsf_1| \Usf)+ H(\Xsf_2| \Usf) = H(X_1, X_2)$.
For the choice of $\Usf$ used to obtain $M_1$, the region of memory over which the two bounds meet is maximized. 
 \end{proof}

\begin{remark}
The Markov chain  $X_1 - U - X_2$ in Theorem \ref{thm:lower bounds coincide} suggests that for the rate triplet $(R_0, R_1, R_2)$$ \,\in \,$$\GWregion(\Usf)$ used in GW-CACM, we require $R_0+R_1+R_2 = H(X_1, X_2)$. The same Markov chain is also used to define Wyner's common information \cite{wyner1975common}. While in Wyner's common information the goal is to minimize $R_0$ subject to $R_0+R_1+R_2 = H(X_1, X_2)$, for $M_1$ in Theorem \ref{thm:lower bounds coincide}, the goal is to maximize $\min(R_1,R_2)$.
\end{remark}

\begin{corollary}
\label{completoverlap} 
If the file library $(X_1,X_2)$ is such that
$$\Big( I(\Xsf_1; \Xsf_2), H(\Xsf_1|\Xsf_2),H(\Xsf_2|\Xsf_1) \Big) \in cl\Big\{ \bigcup \,\GWregion (\Usf)\Big\}, $$
where the  union is over all choices of $\Usf$ with $|\Uc| \leq |\Xc_1|.|\Xc_2| + 2,$  then $\Rlb= \RGWlb$.

\end{corollary}
\begin{proof}
See \cite{ParisaISIT17}.
\end{proof}

\begin{example}\label{ex:cover all}
Consider a 2-DMS whose joint pmf  $p(x_1, x_2)$ is such that $(\Xsf_1,\Xsf_2)$ can be represented as $\Xsf_1 = (\Xsf_1',\Vsf)$ and $\Xsf_2 =(\Xsf'_2,\Vsf)$, where $\Xsf'_1$ and $\Xsf'_2$ are conditionally independent given $\Vsf$. Taking $\Usf= \Vsf$, the point $R_0 =I(\Xsf_1,\Xsf_2;\Usf) = H(\Vsf) = I(\Xsf_1;\Xsf_2)$, $R_1=H(\Xsf_1|\Usf)= H(\Xsf'_1|\Vsf) = H(\Xsf_1|\Xsf_2)$ and $R_2=H(\Xsf_2|\Usf)= H(\Xsf'_2|\Vsf) = H(\Xsf_2|\Xsf_1)$ belongs to the Gray-Wyner rate region. Then, by Theorem  \ref{completoverlap}, $\Rlb= \RGWlb$ for any $M$.  
\end{example}

\section{An achievable scheme based on GW-CACM and its optimality} \label{sec:achievable scheme} 
In this section, we present an achievable GW-CACM scheme
, where $i$) the first step consists of a Gray-Wyner encoder restricted to operate on the  plane of the  Gray-Wyner
rate region with  $R_1=R_2=\private$, and
%
%
$ii$) the second step is a deterministic correlation-unaware multiple-request CACM scheme that combines ideas from conventional caching (e.g., LFU\footnote{LFU is a local caching policy that, in the setting of this paper, leads to all receivers caching the same part of the file.} and uncoded multicasting) 
and correlation-unaware CACM with coded placement, as suggested by Tian and Chen in 
\cite{tian2016caching}, and referred to as TC in the following. 
We remark that jointly optimizing these two steps is the key to maximizing the overall performance. 

We then refer to the overall scheme as GW-LFU-TC, which works as follows:

{\bf Gray-Wyner Encoder:} generates
three library descriptions $\{\Wsf_0, \Wsf_1, \Wsf_2\}$ using a conditional pmf $p(u|x_1,x_2)$ such that $p(x_1|u)=p(x_2|u)$ with $(R_0, \rho, \rho)$$ \,\in \,$$\GWregion(\Usf)$.

 {\bf Cache Encoder:} populates the receiver caches as:
\begin{itemize}
\item If $M\in [0, \private)$, the common description $\Wsf_0$ is not cached at either receiver, and descriptions $\{\Wsf_1, \Wsf_2\}$ are cached according to the caching phase of TC. 
\item If $M\in [\private, R_0+\private)$, the first $n(M-\private)$ bits of description $\Wsf_0$ are cached at both receivers (as per LFU caching), and descriptions $\{\Wsf_1, \Wsf_2\}$ are cached according to TC over the remaining cache capacity $\private$. 

\item If $M\in [R_0+\private, R_0+2\private]$, the common description $\Wsf_0$ is fully cached at both receivers, and descriptions $\{\Wsf_1, \Wsf_2\}$ are cached according to TC over the remaining cache capacity $M-R_0$. 
\end{itemize}

{\bf Multicast Encoder:} transmits the descriptions $\{\Wsf_1, \Wsf_2\}$ according to conventional coded multicast schemes \cite{maddah14fundamental,ji15order,ji2015caching,tian2016caching}, 
while  the portion of  $\Wsf_0$ missing at each receiver cache is transmitted via uncoded (naive) multicast. 

\begin{remark}
Differently from the single-cache setting analyzed in \cite{timo2016rate}, where caching the common description first is always optimal, in our case, when the cache capacity is smaller than the private description size, $\private$, it is optimal to first cache the private descriptions. 
\end{remark}


\begin{theorem}\label{thm:achievable rate}
Given a conditional pmf $p(u|x_1,x_2)$ such that $p(x_1|u)=p(x_2|u)$, a  cache capacity $M$, and a rate triplet  $(R_0, \rho, \rho)$$ \,\in \,$$\GWregion(\Usf)$, 
the peak $\Usf$-rate achieved by GW-LFU-TC is given by

\begin{equation}
\Rach(M,\Usf) = \inf R_{ach}(R_0,\private),
\end{equation}
where the infimum is over all rate triplets $(R_0,\rho,\rho) \in \GWregion(\Usf)$, and    $R_{ach}(R_0,\private)$ is 
\begin{equation}
R_{ach}(R_0,\private) =
\begin{cases}
 R_0 + 2\private  -2 M ,                   & \; M \in [0, \frac{1}{2}\private)  \\
 R_0  + \frac{3}{2}\private -  M,				      &\; M\in [\frac{1}{2}\private, R_0+\private) \\
 \frac{1}{2}R_0  + \private -  \frac{1}{2}M,	    &\; M\in [R_0+\private, R_0+2\private]. \notag
  \end{cases}
\end{equation}
Furthermore, optimizing over $\Usf$, the peak rate achieved by GW-LFU-TC is given by
$$\Rach(M) = \inf\, \Rach(M,\Usf),$$ where the infimum is over all choices of $\Usf$ with $|\Uc| \leq |\Xc_1|.|\Xc_2| + 2.$
\end{theorem}
\begin{proof}
See \cite{ParisaISIT17}.
\end{proof}
 
\subsection{Optimality of GW-LFU-TC}\label{sec:Order Optimality}

In order to prove the optimality of the GW-LFU-TC scheme, we first state the following theorem:

 \begin{theorem}\label{thm:achievable and lower bound} 
For any $\Usf$ and $M$, 
$$\Rach(M,\Usf)=\RGWstarU.$$
\end{theorem}
\begin{proof}
Similar to the proof of Theorem~\ref{thm:lower bounds coincide}, $\Rach(M,\Usf)$ is compared to the lower bound $\RGWUlb$ in each memory region, and for any $\Usf$ forming a Markov chain, $\Xsf_1-\Usf-\Xsf_2$.
\end{proof}
\begin{example}
Assuming the same setting as in Example~\ref{ex:cover all}, since $\Rlb= \RGWlb$ for any $M$, it follows from Theorem \ref{thm:achievable and lower bound} that the GW-LFU-TC scheme is optimal for all values of memory $M$. 
\end{example}
The following theorem characterizes the performance of the GW-LFU-TC scheme for different regions of $M$, and delineates the cache capacity region for which the scheme is optimal or near optimal. 
\begin{theorem}\label{thm:optimal achievable} 
Let 
\begin{align}
& \widetilde{M}_1 \triangleq  \max\limits_{X_1 - U - X_2} \frac{1}{2}  \min \Big\{H(\Xsf_1|\Usf),H(\Xsf_2|\Usf) \Big\}, \notag 
\end{align}
where the $\max$ is over all choices of $\Usf$ such that $p(x_1|u)=p(x_2|u)$. 

Then, for $M \in [0,\;\widetilde{M}_1 ]\cup [ H(\Xsf_1, \Xsf_2)-2\widetilde{M}_1, \;H(\Xsf_1, \Xsf_2) ]$, the GW-LFU-TC  scheme is optimal  i.e., $\Rach(M)= \Rstar$.

\noindent In addition, for $M\in (\widetilde{M}_1, \; H(\Xsf_1, \Xsf_2)-2\widetilde{M}_1)$, we have
$$ \Rach(M) - \Rstar \,\leq\ \frac{1}{2} \min \Big\{H(\Xsf_1|\Xsf_2),H(\Xsf_2|\Xsf_1) \Big\} - \widetilde{M}_1.$$
\end{theorem}

\begin{proof}
See \cite{ParisaISIT17}
\end{proof}

\subsection{Illustration of Results: Doubly Symmetric Binary Source}\label{sec:Example1}
Consider, as a 2-DMS, a doubly symmetric binary source (DSBS) with joint pmf 
$p(x_1,x_2) = \frac{1}{2} (1-p_0)\delta_{x_1,x_2} + \frac{1}{2} p_0(1-\delta_{x_1,x_2})$,  $x_1,x_2 \in\{0,1\}$, 
and parameter $p_0\in[0,\frac{1}{2}]$.
Then,
\begin{align}
& H(\Xsf_1) = H(\Xsf_2) = 1,\notag \\
& H(\Xsf_1|\Xsf_2) = H(\Xsf_2|\Xsf_1) = h(p_0),\notag \\
& H(\Xsf_1,\Xsf_2) = 1 + h(p_0),\notag
\end{align}
where $h(p) = -p \log(p)-(1-p)\log(1-p)$ is the binary entropy function.
As derived in \cite{wyner1975common}, an achievable Gray-Wyner rate region of a DSBS restricted to the plane $\{(R_0,R_1,R_2) : R_1=R_2=\private\}$,  is described by the set of rate triplets  $(R_0,\private,\private)$ with $R_0$ given by
\begin{equation}
R_0 \geq \begin{cases} 1+ h(p_0) - 2\private,              &\hspace{0.2cm} 0 \leq \private < h(p_1) \\
  f(\private) 			                                &\hspace{0.2cm} h(p_1)\leq \private \leq 1
  \end{cases},
\end{equation}
where $p_1=\frac{1}{2}(1-\sqrt{(1-2p_0)})$,
\begin{align}
& f(\private) \triangleq 1+h(p_0)+ \Big(h^{-1}(\private)-\frac{p_0}{2} \Big)\log\Big(h^{-1}(\private)-\frac{p_0}{2} \Big) +\notag \\
& p_0 \log\Big(\frac{p_0}{2}\Big) + \Big( 1-h^{-1}(\private)-\frac{p_0}{2} \Big)\log\Big( 1-h^{-1}(\private)-\frac{p_0}{2} \Big),\notag
\end{align}
and $h^{-1}(\private)$ is the inverse of the binary entropy function.  
 
We compare the performance of the proposed GW-LFU-TC scheme with respect to: 1) LFU caching with uncoded multicasting (LFU-UM), 2) the deterministic correlation-unaware CACM in \cite{tian2016caching}, referred to as TC, 3) the lower bound on the GW-CACM peak rate-memory function ($R^{LB}_{GW}$), and 4) the lower bound on the optimal peak rate-memory function ($R^{LB}$). Fig.~\ref{fig:simulations}  displays  the rate-memory trade-offs for $p_0 = 0.2$. 

In line with Theorems \ref{thm:lower bounds coincide} and \ref{thm:optimal achievable}, Fig.~\ref{fig:simulations} shows that the lower bound on the Gray-Wyner rate-memory function ($R^{LB}_{GW}$) coincides with the lower bound on the optimal rate-memory function ($R^{LB}$) for $M \leq \widetilde{M}_1=0.25$ and $M\geq ( H(\Xsf_1,\Xsf_2) -2\widetilde{M}_1) =1.22$, and GW-LFU-TC is optimal in this region, while correlation-unaware schemes, LFU-UM and TC, fall short. Furthermore, the gap between the rate achieved with GW-LFU-TC and the optimal peak rate-memory function is less than $0.11$, which is less than half of the conditional entropy, $0.36$. Finally, in line with Theorem \ref{thm:achievable and lower bound}, GW-LFU-TC achieves $R^{LB}_{GW}$ for any $M$.


\begin{figure} 
\centering
\includegraphics[width=0.82\linewidth]{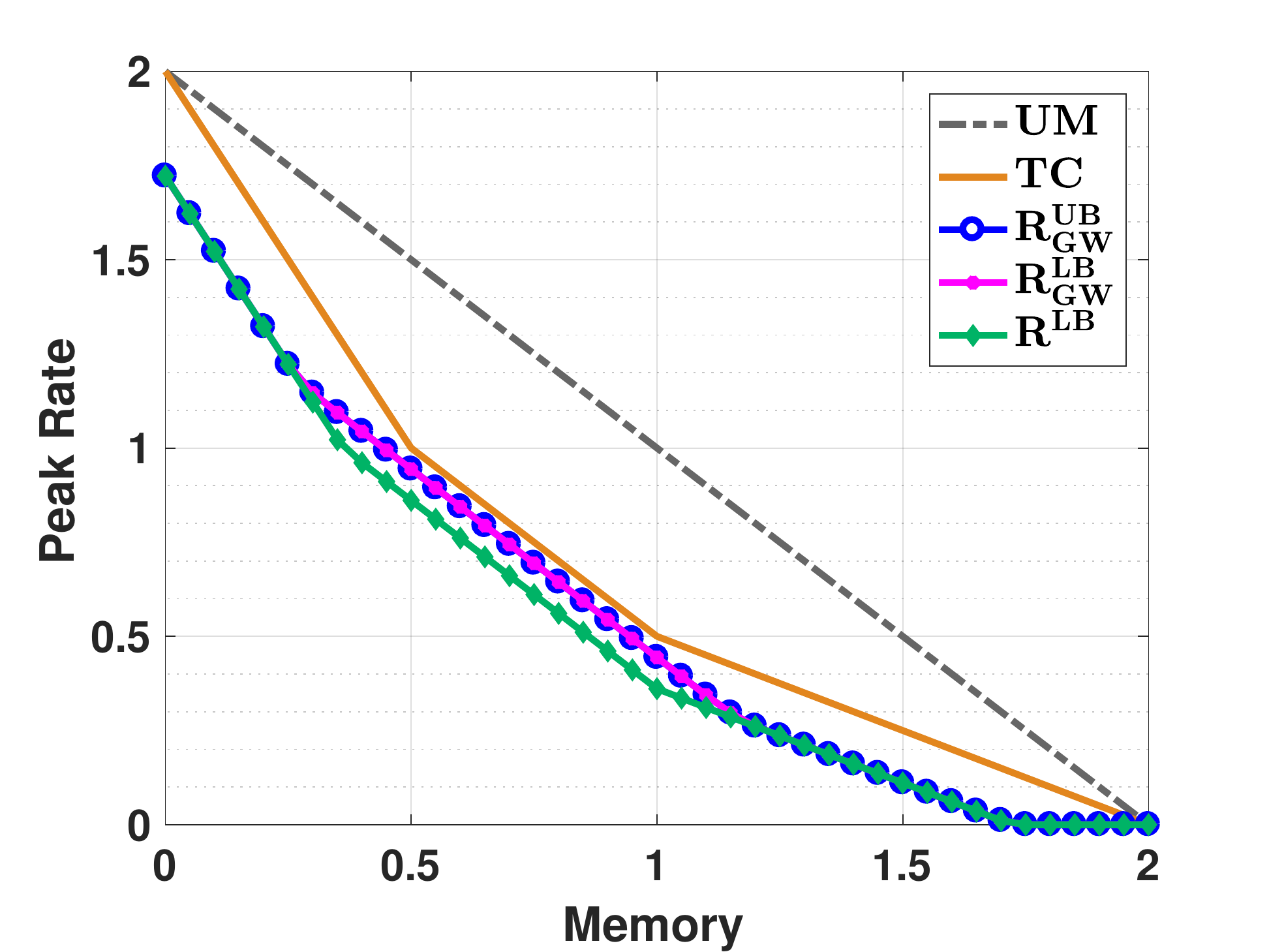}
\caption{Rate-memory trade-off for a DSBS with $p_0 = 0.2$.}
\label{fig:simulations}
\end{figure}

\section{Concluding Remarks}\label{sec:Conclusions}
In this paper, we have studied the fundamental limits of cache-aided communication systems under the assumption of correlated content for a two-user two-file network.
We have derived a lower bound on the peak rate-memory function for such systems
and proposed a class of schemes based on
a two-step source coding approach.
Files are first compressed using 
Gray-Wyner source coding, 
and then cached and delivered using a combination of existing correlation-unaware cached aided coded multicast schemes. 
We have fully characterized the rate-memory trade-off of such class of schemes, proposed an achievable two-step scheme, and proved its optimality for different memory regimes. 
Finally, in \cite{ParisaISIT17}, we provide an extended analysis that includes the characterization of both peak and average rate-memory trade-offs in more general user-file settings. 

\section*{Acknowledgement}
The authors would like to thank M. Wigger and D. G$\ddot{\text u}$nd$\ddot{\text u}$z  for their useful discussions on the Gray-Wyner network.
\bibliographystyle{IEEEtran}
\bibliography{References2}

\end{document}